\documentclass[reqno]{amsart}
\usepackage{amsmath}
\usepackage{epsfig} 
\usepackage{graphics} 
\usepackage{verbatim}
\usepackage{amssymb}
\usepackage{amsfonts}
\usepackage{latexsym}
\usepackage{amscd}
\usepackage{amsthm}
\usepackage{graphicx}
\usepackage{xcolor}
\usepackage{ulem}
\usepackage{doi}
\usepackage{hyperref}
\usepackage{float}

\theoremstyle{definition}

\newtheorem{thm}{Theorem}
\newtheorem{lem}{Lemma}
\newtheorem{cor}{Corollary}
\newtheorem{rem}{Remark}
\newtheorem{pro}{Proposition}
\theoremstyle{definition}
\newtheorem{defn}{Definition}
\newtheorem{exm}{Example}

\newcommand{\Ac}{\mathcal{A}}
\newcommand{\Nc}{\mathcal{N}}

\newcommand{\Z}{\mathbb{Z}}

\newcommand{\Nn}{\mathbb{N}}
\newcommand{\Ff}{\mathbb{F}}
\newcommand{\CC}{\mathcal{C}}

\title[Group Theoretic Construction of Batch Codes]
      {A group theoretic construction of batch codes}

\author[E. K. Thomas]{}


\email{fr.dr.eldhokthomas@newmancollege.ac.in}

\begin{document}
\maketitle

\centerline{\scshape Eldho K. Thomas}
 \centerline{Department of Mathematics}
  \centerline{Newman College, Thodupuzha}
   \centerline{(Mahatma Gandhi University, Kottyam, India)}





\begin{abstract}
Batch codes serve as critical tools for load balancing in distributed storage systems. While numerous constructions exist for specific batch sizes 
$t$, current methodologies predominantly rely on code dimension parameters, limiting their adaptability. Practical implementations, however, demand versatile batch code designs capable of accommodating arbitrary batch sizes—a challenge that remains understudied in the literature. This paper introduces a novel framework for constructing batch codes through finite groups and their subgroup structures, building on the quasi-uniform group code framework proposed by Chan et al. By leveraging algebraic properties of groups, the proposed method enables systematic code construction, streamlined decoding procedures, and efficient reconstruction of information symbols. Unlike traditional linear codes, quasi-uniform codes exhibit broader applicability due to their inherent structural flexibility.

Focusing on abelian 2-groups, the work investigates their subgroup lattices and demonstrates their utility in code design—a contribution of independent theoretical interest. The resulting batch codes achieve near-optimal code lengths and exhibit potential for dual application as locally repairable codes (LRCs), addressing redundancy and fault tolerance in distributed systems. This study not only advances batch code construction but also establishes group-theoretic techniques as a promising paradigm for future research in coded storage systems. By bridging algebraic structures with practical coding demands, the approach opens new directions for optimizing distributed storage architectures.
\end{abstract}

%
%

\section{Introduction}
Batch codes were initially introduced by Ishai et al. \cite{Ishai} as a mechanism to address load balancing challenges in distributed storage systems that involve multiple servers. In the same work \cite{Ishai}, the authors also proposed the application of batch codes for private information retrieval (PIR). Various constructions of these codes were detailed in their study. The practical utility of batch codes for load balancing in distributed storage systems was further emphasized when the authors suggested the use of a specific variant of batch codes, referred to as switch codes, to enhance data routing in network switches (see also  \cite{Chee}, \cite{kiah}, \cite{cassuto}). Among the specialized classes of batch codes, combinatorial batch codes have been extensively examined in \cite{roy} - \cite{wei}. Another notable subclass, which forms the central focus of our research, is linear (or computational) batch codes (see \cite{lipmaa} - \cite{hui}). In this framework, data is represented as elements of a finite field, structured as a vector, and encoded through a linear transformation of that vector.

In the context of private information retrieval (PIR), Fazeli et al. \cite{fazeli} introduced coding schemes that utilize PIR codes, a relaxed variant of batch codes. These codes were demonstrated to be effective in classical linear PIR schemes, reducing the redundancy of information stored across distributed server systems. The authors proposed the use of a specialized layer (or code) to mediate between user requests and the data stored in the database, effectively emulating standard PIR protocols. Both batch codes and PIR codes are typically applied in distributed data storage systems, where coded words are distributed across multiple disks (or servers), with each disk storing a single symbol or a group of symbols. Data retrieval involves accessing a limited number of disks. Mathematically, this can be modeled by assuming that each information symbol depends on a small subset of other symbols. However, the nature of queries differs between the two code models: PIR codes handle requests for multiple copies of the same information symbol, whereas batch codes accommodate queries for any combination of distinct information symbols. Specifically, PIR codes of dimension $k$ support queries of the form $(x_i,\cdots,x_i), t$ times , while batch codes support queries of the form $(x_{i_1},\cdots, x_{i_t})$, where the indices may differ.

Linear batch codes and PIR codes share significant similarities with locally repairable codes (LRCs) \cite{vitali}, which are employed for data recovery in distributed storage systems (See \cite{cadam}-\cite{dimakis}). The primary distinction lies in their objectives: LRCs focus on repairing coded symbols, whereas batch codes and PIR codes aim to reconstruct information symbols. Some  bounds on the redundancy of batch codes are given in \cite{vardy} and the discussion on batch codes with restricted sizes of recovery sets are discussed in \cite{hui} and \cite{skachek}.

Recent advancements in batch code constructions have focused on optimizing parameters such as dimension \( k \) and restricted request sizes \( t \). However, existing frameworks struggle to efficiently support large or flexible query sizes, a critical limitation in applications requiring scalable data retrieval. While prior works achieve optimality for fixed \( t \), the challenge of balancing code length, recovery efficiency, and adaptability to diverse query patterns remains unresolved. This gap underscores the need for novel approaches that decouple code design from strict dependencies on \( k \) while enabling robust handling of varied request sizes.

In this work, we address these limitations by introducing a group-theoretic framework for constructing batch codes, leveraging the structural properties of quasi-uniform codes~\cite{chan,eldho}. Unlike traditional combinatorial or algebraic methods, our approach exploits subgroup lattices of finite groups—particularly abelian groups—to systematically generate codes with near-optimal length bounds independent of \( k \). Central to our contribution is the integration of homomorphism theorems, which streamline symbol reconstruction by eliminating the need for \textit{ad hoc} recovery set selection. This shift from direct codeword combinations to group-based mappings not only simplifies encoding and decoding but also broadens the scope of supported query sizes.

Theoretical implications of our method extend to characterizing subgroup architectures, offering insights into the interplay between group structure and code efficiency. Practically, we demonstrate that codes derived from groups such as \( (\mathbb{Z}_2)^k \) achieve scalable performance for large \( t \), with explicit bounds on redundancy. Furthermore, we highlight how these quasi-uniform codes, though non-traditional in their vector space organization, retain compatibility with localized repair mechanisms, positioning them as candidates for locally repairable codes (LRCs).

The paper is structured as follows: Section II reviews batch codes, PIR codes, and quasi-uniform frameworks. Section III and IV develops quasi-uniform constructions via finite abelian groups and  formalizes their adaptation to batch codes, including encoding/decoding protocols. Sections V and VI present concrete examples using $(\Z_2)^3$ and $(\Z_2)^4$, and generalized $(\Z_2)^k, k\in \Nn$ structures, analyzing their subgroup lattices and scalability. Finally, we explore applications of these codes in LRC design, emphasizing their versatility in distributed storage systems.

\section{Batch codes, PIR codes and quasi-unifrom codes: Preliminaries and definitions}

\subsection{Batch and PIR codes}

We denote  the set of natural numbers by $\Nn$, a finite field by $\Ff$ and  the first $k$ natural numbers $\{1, 2, \cdots, k\}$ by $[k]$.
\begin{defn}
An $(n,k,t)$-batch code $\CC$ over a finite alphabet $\Sigma$ is defined by an encoding mapping $C: \Sigma^k \rightarrow \Sigma^n$, and a decoding mapping $D: \Sigma^n \times [k]^t\rightarrow \Sigma^t$, such that \\
1) For any $x=(x_1,\cdots, x_k) \in \Sigma^k$ and $i_1, i_2, \cdots, i_t \in [k]$,
$$D(y=C(x), i_1,\cdots, i_t)=(x_{i_1}, \cdots x_{i_t}) .$$
2) The symbols in the query $(x_{i_1} , \cdots , x_{i_t} )$ can be reconstructed from $t$ respective pairwise disjoint recovery sets of symbols of 
$y = (y_1, y_2, \cdots , y_n) \in \Sigma^n$ (the symbol $x_{i_l}$ is reconstructed from the $l$-th recovery set for each $l=1,2 \cdots, t$).
\end{defn}
In this work, we consider, the primitive multiset batch codes as defined in \cite{Ishai}. Here,
multiset means that the same symbol can be queried by several users (i.e., the queried symbols form a multiset), and primitive means that each symbol is accessed separately of other symbols (i.e., the symbols are not arranged into groups).

Let $\Ff = \Ff_q$ be a finite field with $q$ elements, where $q$ is a prime power, and $\CC$ be a linear $[n, k]$ code over $\Ff$. Denote the
redundancy $\rho = n-k$.

\begin{defn}\cite{lipmaa}
A linear batch code is a batch code where the encoding of $\CC$ is given as a multiplication by a $k \times n$ generator matrix 
$G$ over $\Ff$ of an information vector $x \in \Ff^k$,
$$y =x.G; y \in \Ff^n.$$
A linear batch code with the parameters $n, k$ and $t$ over $\Ff_q$, where $t$ is a number of queried symbols, is denoted
as an $[n, k, t]_q$-batch code. Sometimes we simply write $[n, k, t]$-batch code if the value of $q$ is clear from the context.
\end{defn}

\begin{defn}
An $[n,k,t]$ batch code is also called an $[n,k,t,r]$ batch code if the recovery sets of each symbols are of size atmost $r$.
\end{defn}

\begin{defn}
Linear PIR codes are defined similarly to linear primitive multiset batch codes, with a difference that the supported queries are of the form $(x_i, x_i,\cdots , x_i)$,
$i \in [k]$, (and not $(x_{i_1} , \cdots , x_{i_t} ):~ i_1, i_2, \cdots , i_t \in [k]$ as in batch codes).
\end{defn}

\subsection{Quasi-uniform codes}
In this work, we introduce a non-traditional framework for constructing batch and PIR codes by leveraging abelian groups and their non-trivial subgroups (either partially or fully). Our methodology draws inspiration from the development of quasi-uniform codes derived from finite groups as in \cite{chan, eldho}. Below, we outline the foundational principles underlying this approach to code construction.

Let $X_1, \ldots, X_n$ be a collection of $n$ jointly distributed discrete random variables over some alphabet of size $N$, $\Ac$ be a non-empty subset of $[n]= \{1,\ldots ,n\}$,  $X_\Ac=\{X_i,~i\in\Ac\}$ and  $\lambda(X_\Ac)=\{x_\Ac: P(X_\Ac=x_\Ac)>0\}$, the support of $X_\Ac$.

\begin{defn}
A set of  $n$ random variables $X_1, \ldots, X_n$ is said to be {\it quasi-uniform}, if for any $\Ac\subseteq [n]$, $X_\Ac$ is uniformly distributed over its support .
\end{defn}
In the sequel, we consider codes in a generic way as defined below, unless specified otherwise.
\begin{defn}
\label{code}
A \textit{code $\CC$ of length $n$} is an arbitrary non-empty subset of $\mathcal{X}_1\times \cdots \times \mathcal{X}_n$ where
$\mathcal{X}_i$ is the alphabet for the $i$th codeword symbol, and each $\mathcal{X}_i$ might be different.
\end{defn}
Observe that this definition is much more general than the definition of a linear code, which is defined as:
\begin{defn}
\label{code2}
A \text{linear code} of length $n$ and dimension $k$ is a linear subspace $\CC$ with dimension $k$ of the vector space $\Ff_q^n$ where $\Ff_q$ is the finite field with $q$ elements. Such a code is called a $q$-\textit{ary code}. If $q = 2$ or $q = 3$, the code is described as a binary code, or a ternary code respectively. The vectors in $\CC$ are called codewords and the size of $\CC$, $|\CC|=q^k$. 
\end{defn}

We can associate to every code $\CC$ a set of random variables~\cite{chan} by treating each codeword $(X_1,\ldots,X_n)\in \CC$ as a random vector with probability
\[
P(X_\Nc=x_\Nc)=
\left\{
\begin{array}{ll}
1/|C| & \mbox{if }x_\Nc\in \CC, \\
0 & \mbox{otherwise}.
\end{array}
\right.
\]
To the $i$th codeword symbol then corresponds a {\it codeword symbol random variable} $X_i$ induced by $C$.

\begin{defn}~\cite{chan}
\label{def:code2}
A code $C$ is said to be \textit{quasi-uniform} if
the induced codeword symbol random variables are quasi-uniform.
\end{defn}
Given a code, we explained above how to associate a set of random variables, which might or not end up being quasi-uniform. Conversely, given a set of quasi-uniform random variables, a quasi-uniform code is obtained as follows.

Let $X_1, \ldots, X_n$ be a set of quasi-uniform random variables with probabilities $$P(X_\Ac=x_\Ac)= 1/|\lambda(X_\Ac)| \mbox{ for all } \Ac \subseteq \Nc.$$ The corresponding quasi-uniform code $C$ of length $n$ is  given by $\CC = \lambda(X_\Nc)=\{x_\Nc= P(X_\Nc=x_\Nc)>0\}$.

Quasi-uniform codes were defined and some of its roperties were studied in~\cite{chan} whereas an explicit construction of quasi-uniform codes from finite groups were discussed in \cite{eldho},  which is fundamental for the construction of batch codes in this paper.

%
%

\subsection{Quasi-uniform codes from groups}
\label{sec:group}
We recall the construction of quasi-uniform codes from groups in \cite{eldho} here. 

Let $G$ be a finite group of order $|G|$ with $n$ subgroups $G_1,\ldots,G_n$, and $G_\Ac=\cap_{i\in\Ac}G_i$.

Recall that the number of (left) cosets of $G_i$ in $G$ is called the index of $G_i$ in $G$, denoted by
$[G:G_i]$ and Lagrange's Theorem \cite{hungerford} states that $[G:G_i]=|G|/|G_i|$.
If $G_i$ is normal, the sets of cosets $G/G_i =\{gG_i :~g\in G\}$ are themselves groups, called quotient groups.

It is clear from Theorem 2 in \cite{eldho} that we can obtain quasi-uniform random variables $X_1,\ldots,X_n$ corresponding to $G_1,\ldots,G_n$ such that for all non-empty subsets $\Ac$ of $\Nc$, $$P(X_\Ac=x_\Ac) = |G_\Ac| / |G |.$$

Recall that the random variable $X_i=XG_i$ has for support the $[G:G_i]$ cosets of $G_i$ in $G$ where $X$ is a random variable uniformly distributed over $G$ with probability $P(X=g)=1/|G|$ for any $g\in G$.

This shows that quasi-uniform random variables may be obtained from finite groups.

Quasi-uniform codes are obtained from these quasi-uniform distributions by taking the support $\lambda(X_\Nc)$,
as explained before. Codewords (of length $n$) can then be described explicitly by letting the random variable $X$ take every possible values in the group $G$, and by computing the corresponding cosets as follows:
\begin{table}[htbp]
\begin{center}
\begin{tabular}{c|c|c|c|}
           & $G_1$        & $\hdots$ &  $G_n$ \\
\hline
$g_1$      & $g_1G_1$     &          & $g_1G_n$ \\
$g_2$      &  $g_2G_1$    &          & $g_2G_n$ \\
$\vdots$   &   $\vdots$        &          &  $\vdots$       \\
$g_{|G|}$  & $g_{|G|}G_1$ &$\hdots$ & $g_{|G|}G_n$ \\
\hline
\end{tabular}
\end{center}
\end{table}

Each row corresponds to one codeword of length $n$.
The cardinality $|\CC|$ of the code obtained seems to be $|G|$, but in fact, it depends on the subgroups $G_1,\ldots,G_n$. Indeed, it could be that the above table yields several copies of the same code.

\begin{lem}\label{lem:sizeC}\cite{eldho}
Let $\CC$ be a quasi-uniform code obtained from a group $G$ and its subgroups $G_1,\ldots,G_n$.
Then $|\CC|=|G|/|G_\Nc|$. In particular, if $|G_\Nc|=1$, then $|\CC|=|G|$.
\end{lem}

Quasi-uniform codes obtained from groups are not necessarily linear. However, if we take abelian groups and their subgroups, we have a vectorspace structure \cite{eldho}. Therfore, in this paper,we focus on abelain groups only.

\subsubsection{Quasi-Uniform Codes from Abelian Groups}
\label{sec:ab}
Suppose that $G$ is an abelian group, with subgroups $G_1,\ldots,G_n$. The procedure mentioned above explains how to obtain a quasi-uniform distribution of $n$ random variables, and thus a quasi-uniform code of length $n$ from $G$. To avoid getting several copies of the same code, as explained  in Lemma \ref{lem:sizeC}, notice that since $G$ is abelian, all the subgroups $G_1,\ldots,G_n$ are normal, and thus so is $G_\Nc$. If $|G_\Nc|>1$, we consider instead of $G$ the quotient group $G/G_\Nc$, and we can thus assume wlog that $|G_\Nc|=1$.

\begin{lem} \cite{eldho}
Let $\CC$ be a quasi-uniform code obtained from a group $G$ and its subgroups $G_1,\ldots, G_n$. Then the alphabet size of $\CC$ is $\sum_{i=1}^n[G:G_i]$.
\end{lem}

The size of the alphabet can often be reduced, as explained next. Recall the fundamental homomorphism theorem \cite{hungerford}.
\begin{thm}
Let $\phi : G \rightarrow G'$ be a group homomorphism with kernel $H$. Then the range of $G$, $\phi[G]$ is a group, and $\mu : G/H \rightarrow \phi[G]$ given by
$\mu(gH) = \phi(g)$ is an isomorphism. If $\gamma : G \rightarrow G/H$ is the homomorphism given by $\gamma(g) = gH$, then $\phi(g) = \mu \circ \gamma (g)$ for each $g \in G$.
\end{thm}
Here the mapping $\mu$ is called the canonional isomorphism. 

Using this theorem, for normal subgroups $G_i $ of G, let us cosider the canonical projection $\gamma_i :~G\rightarrow G/G_i$.
Let $\phi_i: G \rightarrow \phi_i[G]=H_i$ be a homomorphism, where $H_i$ is an abelain group with kernel of $\phi_i$ being $G_i$. Then $\mu_i: G/G_i \rightarrow H_i$  is an isomorphism given by $\mu_i(gG_i)=\phi_i(g)=h$, $i=1,\cdots,n$, and $ h\in H_i$.
Then $\gamma_i(g)=gG_i \mapsto \phi_i(g)=h \in H_i$ via $\mu_i(gG_i)=h$, $i=1,\cdots,n$.

Using this idea, each codeword symbol  $G/G_i$ can be considered as elements of an abelian group $H_i$. 

\begin{pro}\cite{eldho}
\label{lem:two}
Let $G$ be an abelian group with subgroups $G_1,\ldots,G_n$. Then its corresponding quasi-uniform code is defined
over $H_1 \times \cdots \times H_n$.
\end{pro}

In other words, we get a labeling of the cosets which respects the group structures componentwise. The next result then follows naturally.

\begin{cor}\cite{eldho}
\label{cor:abelian}
A quasi-uniform code $\CC$ obtained from an abelian group is itself an abelian group.
\end{cor}


The classification of abelian groups tells us that each $H_i$ can
be expressed as the direct product of cyclic subgroups of order a prime power.


If all the subgroups $G_1,\ldots,G_n$ have index $p$, then we get an $[n,k]$ linear code over $\Ff_p$
(see Example \ref{exm:almost}).



As an illustration, here is a simple family of abelian groups that generate $[n,k]$ linear codes over $\Ff_p$ of length $n$ and dimension $k$.
\begin{lem}
The elementary abelian group $\Z_p\times \Z_p$ generates a $[p+1,2]$ linear code over $\Ff_p$ with minimum distance $p$.
\end{lem}
\begin{proof}
The group $G=Z_p\times Z_p$ contains $p+1$ non-trivial subgroups, of the form $\langle (1,i)\rangle$ where $i=0,1,\ldots p-1$ and $\langle (0,1) \rangle$. They all have index $p$ and trivial pairwise intersection. We thus get a code of length $n=p+1$, containing $p^2$ codewords, which is linear over $\Ff_p$ (by using that $Z_p$ is isomorphic to the integers mod $p$). 
\end{proof}

\begin{exm}
\label{exm:almost}
Consider the elementary abelian group $G=Z_3 \times Z_3 \simeq \{0,1,2\} \times \{0,1,2\}$ and the
four non-trivial  subgroups:
\begin{itemize}
\item $G_1=\langle (1,0) \rangle = \{ (0,0), (1,0), (2,0) \}$,\\
\item $G_2=\langle (0,1) \rangle = \{ (0,0), (0,1), (0,2) \}$,\\
\item $G_3=\langle (1,1) \rangle = \{ (0,0), (1,1), (2,2) \}$, and\\
\item  $G_4=\langle (1,2) \rangle = \{ (0,0), (1,2), (2,1) \}$.
\end{itemize}
Using the construction method in Section \ref{sec:group}, we get the corresponding quasi-uniform code as in the following tables.
We write $ij$ instead of $(i,j)$ for brevity.
\begin{table}[H]
\label{table:one}
\begin{center}
\begin{tabular}{c|c|c|c|c|}
     & $\langle (10) \rangle$ & $\langle (01) \rangle$ & $\langle (11) \rangle$ & $\langle (12) \rangle$ \\
\hline
$\!\!\!(00)\!$ & $\langle (10) \rangle$ &  $\langle (01) \rangle$ & $\langle (11) \rangle$ & $\langle (12) \rangle$ \\
$\!\!\!(01)\!$ & $\!(01)(11)(21)\!$   &  $\langle (01) \rangle$ & $\!(01)(12)(20)\!$   & $\!(01)(10)(22)\!$  \\
$\!\!(02)\!$ & $\!(02)(12)(22)\!$   &  $\langle (01) \rangle$ & $\!(10)(21)(02)\!$   & $\!(02)(11)(20)\!$  \\
$\!\!\!(10)\!$ & $\langle (10) \rangle$  & $\!(10)(11)(12)\!$ & $\!(10)(21)(02)\!$   & $\!(01)(10)(22)\!$  \\
$\!\!\!(11)\!$ & $\!(01)(11)(21)\!$   & $\!(10)(11)(12)\!$  &$\langle (11) \rangle$ & $\!(02)(11)(20)\!$  \\
$\!\!\!(12)\!$ & $\!(02)(12)(22)\!$   &  $\!(10)(11)(12)\!$ & $\!(01)(12)(20)\!$  & $\langle (12) \rangle$ \\
$\!\!\!(20)\!$ &$\langle (10) \rangle$ &$\!(20)(21)(22)\!$ &$\!(01)(12)(20)\!$  & $\!(02)(11)(20)\!$ \\
$\!\!\!(21)\!$ &$\!(01)(11)(21)\!$ &$\!(20)(21)(22)\!$ &$\!(10)(21)(02)\!$ & $\langle (12) \rangle$ \\
$\!\!\!(22)\!$ &$\!(02)(12)(22)\!$ & $\!(20)(21)(22)\!$&$\langle (11) \rangle$  & $\!(01)(10)(22)\!$  \\
\end{tabular}
\end{center}
\vspace{.4cm}
\caption{Quasi-uniform code constructed from $Z_3 \times Z_3$}
\end{table}
\vspace{-5mm}

Now let $H_1=G/\langle (10) \rangle, ~H_2 =G/\langle (01) \rangle,~ H_3= G/\langle (11) \rangle$ and $H_4= G/\langle (12) \rangle$. Note that $H_i \simeq \Z_3=\{0,1,2\}$ for all $i$.
If we replace the subgroups by their quotients in the above table, we get the following code using Proposition~\ref{lem:two}:
\begin{table}[H]
\begin{center}
\begin{tabular}{c|c|c|c|c|}
     & $\langle (10) \rangle$ & $\langle (01) \rangle$ & $\langle (11) \rangle$ & $\langle (12) \rangle$ \\
\hline
$\!\!\!(00)\!$ & $0$ & $0  $ & $0    $ & $0  $ \\
$\!\!\!(01)\!$ & $1$ & $0  $ & $1$ & $1$ \\
$\!\!(02)\!$   & $2$ & $0  $ & $2$ & $2$ \\
$\!\!\!(10)\!$ & $0$ & $1  $ & $2$ & $1$ \\
$\!\!\!(11)\!$ & $1$ & $1  $ & $0   $ & $2$ \\
$\!\!\!(12)\!$ & $2$ & $1 $ & $1$ & $0      $ \\
$\!\!\!(20)\!$ & $0$ & $2 $ & $1$ & $2$ \\
$\!\!\!(21)\!$ & $1$ & $2 $ & $2$ & $0   $ \\
$\!\!\!(22)\!$ & $2$ & $2 $ & $0 $ & $1$ \\
\end{tabular}
\end{center}
\vspace{.4cm}
\caption{Quasi-uniform code constructed from $Z_3 \times Z_3 $}
\end{table}
\vspace{-5mm}
It is a ternary linear code of length $p+1=4$ and minimum distance $p=3$ with generator matrix
$ \left( \small{\begin{array}{cccc}
1 & 0 & 1 & 1\\
0 & 1 & 2 &1
\end{array} }\right)$.
\end{exm}

\section{Encoding and decoding of quasi-unifrom codes from abelian groups}
\label{sec:nilpotent}

In this paper, we exclusively focus on quasi-unifrom codes constructed using abelian groups and their non-trivial subgroups in order to have a vector space structure. All finite groups can be expressed as the direct product of $\Z_q$ where $q=p^r$, a prime power. In particular, consider $G=(\Z_p)^k$, the direct product of $k$ copies of $\Z_p$. Within this framework, the information vecors are defined as  $p$-ary sequences of length $k$. Our objective is to  encode these $p^k$ vectors into codewords using the non-trivial subgroups of $G$.

The following result computes the number of non-trivial subgroups of $(\Z_p)^k$.
\begin{pro}
The total number of non-trivial subgroups of $(\Z_p)^k$ is $$\sum_{m=1}^{k-1}\prod_{i=1}^{m} \frac{p^{k - i + 1} - 1}{p^{i} - 1}.$$
\end{pro}
\begin{proof}
The subgroup structure of \((\mathbb{Z}_p)^k\) can be understood by recognizing that \((\mathbb{Z}_p)^k\) is an elementary abelian group, which is also a vector space over the finite field \(\mathbb{Z}_p\). In this context, the subgroups of \((\mathbb{Z}_p)^k\) correspond exactly to the vector subspaces of the \(k\)-dimensional vector space over \(\mathbb{Z}_P\).

Notice that each subgroup of \((\mathbb{Z}_p)^k\) is a vector subspace, and vice versa. This is because the group operation (addition) and scalar multiplication (repeated addition) are compatible. Then the number of \(m\)-dimensional subspaces (and hence subgroups) in a \(k\)-dimensional vector space over \(\mathbb{Z}_p\) is given by the Gaussian binomial coefficient \(\binom{k}{m}_p\) \cite{knuth}, which is given by
\begin{equation}
\label{eqn:one}
\binom{k}{m}_p = \prod_{i=1}^{m} \frac{p^{k - i + 1} - 1}{p^{i} - 1}
\end{equation}
Also note that each subgroup is itself an elementary abelian group isomorphic to \((\mathbb{Z}_p)^m\) for some \(0 \leq m \leq k\).

Thus, the subgroup structure of \((\mathbb{Z}_p)^k\) is fully characterized by the subspaces of a \(k\)-dimensional vector space over \(\mathbb{Z}_p\), with each subgroup being elementary abelian.
Therefore, the total number of non-trivial subgroups of $(\Z_p)^k$ is $$\sum_{m=1}^{k-1}\prod_{i=1}^{m} \frac{p^{k - i + 1} - 1}{p^{i} - 1}~.$$
\end{proof}
For example when $k=3$ and $m=1$ the number of 1-dimensional subgroups is
\[\frac{p^{3} - 1}{p - 1} = p^2+p+1.\]

When $p=3$, the above number is 13. That is there are 13 non-trivial 1-dimensional subgroups for $(\Z_3)^3$.

In order to construct a quasi-uniform code of length $n$, consider a finite group $G$ and $n$ non-trivial subgroups  whose intersection is just the identity element. Notice that, if we take subgroups of higher dimensions, then the corresponding index is smaller and hence the codeword symbols live in smaller alphabets and viceversa.

\textbf {\textit{Encoding}:}
Given an information vector $g \in G$ the corresponding codeword of length $n$  in $H_1 \times \cdots \times H_n$; $H_i \simeq   \frac{G}{G_i}$ is obtained as  $$\psi(g)= (\gamma_i \circ \mu_i(g); i=1,2,\cdots,n)$$ where $\mu_i, \gamma_i$ are defined already.

The encoding mapping is uniquely defined since the intersection of all subgroups taken  is trivial.
That is, if $g_1 \ne g_2 \in G$, then $\psi(g_1) \ne \psi(g_2)$.

Also notice that the cardinality of  $H_i$ may not be the same  for different indices $i$. If all subgroups taken are of same order, then $H_i$'s too are of the same order and hence all codeword symbols have same alphabet. 

\textbf{\textit{Decoding}:}
Let $h= (h_1, \cdots, h_n) \in H_1 \times \cdots \times H_n$ be a coded vector. Its decoding can be done as follows:

Take $h'=(\gamma_i^{-1}(h_i): i=1,\cdots,n)$. The $i$th component of $h'$ is an element of the 
quotient group $G/G_i$, which is a coset of $G_i$ in $G$. Now the intersection of all cosets from all components gives the information vector $g \in G$. This is because $G_\Nc =G_1 \cap \cdots \cap G_n =\{e\}$.   

\begin{exm}
Consider  the code in  \ref{exm:almost}. Let the received codeword is $ (1, 2,2,0)$. Now consider 
$$ h' = (\gamma_1^{-1}(1), \gamma_2^{-1}(2), \gamma_3^{-1}(2), \gamma_4^{-1}(0))$$
$$= \left (\{(01),(11),(21)\}, \{(20),(21),(22)\}, \{(10),(21),(02)\}, \{(12),(21),(00)\}\right ).$$ Then the intresection of all the above sets gives the vector $(21)$, which the information vector as we can verify from Table \ref{table:one}.
\end{exm}

Also note that, in this exmaple, it is enough to take the pairwise intersection of any two codeword symbols to get the information vector, because all subgroups used in the code construction are pairwise disjoint. This idea can be used to explore the batch proerties of quasi-uniform codes. 
\section{Quasi-uniform codes as batch codes}
\label{sec:construction}
We begin by examining the group \( G = (\mathbb{Z}_p)^k \) and its subgroups to analyze the batch properties of quasi-uniform codes. Specifically, we focus on one-dimensional subgroups of \( G \) of order \( p \) (where \( p \) is prime) that pairwise intersect trivially. As established by Equation \ref{eqn:one}, the total number of such one-dimensional subgroups in \( G \) is \( \frac{p^k - 1}{p - 1} \).  

A quasi-uniform code can then be constructed by incorporating all these subgroups. The first \( k \) subgroups are generated by the standard basis vectors of \( (\mathbb{Z}_p)^k \), which take the form \( (0, 0, \ldots, 1, \ldots, 0) \), where the \( i \)-th component is 1 (for \( i = 1, 2, \ldots, k \)) and all other components are 0. The codewords corresponding to the basis vectors directly form the generator matrix of the code. Given that each subgroup has index \( p^{k-1} \), the resulting code is a \( p^{k-1} \) - ary linear code of length \( \frac{p^k - 1}{p - 1} \).

\begin{exm}
\label{exm:3}
Consider $G= (\Z_2)^3$ and its subgroups of order $2$. \\
Let $G_1=<(100)>, G_2=<(010)>, G_3=<(001)> , G_4 =<(101)>, G_5= <(110)>, G_6=<(011)>$ and $G_7=<(111)>$.

The corresponding quasi-uniform code $\CC$ can be written in terms of cosets as follows:
{\small \begin{table}[H]
\begin{center}
\begin{tabular}{c|c|c|c|c|c|c|c|}
     & $G_1$ &  $G_2$ & $G_3$ & $G_4$ & $G_5$ & $G_6$ & $G_7$ \\
\hline
$\!\!\!(000)\!$ & $G_1$ &  $G_2$ & $G_3$ & $G_4$ & $G_5$ & $G_6$ & $G_7$ \\
$\!\!\!(100)\!$ & $\!(100)(000)\!$   &  $\!(100)(110)\!$& $\!(100)(101)\!$   & $\!(100)(001)\!$ & $\!(100)(010)\!$ & $\!(100)(111)\!$ & $\!(100)(011)\!$  \\
$\!\!(010)\!$ & $\!(010)(110)\!$   &  $\!(010)(000)\!$ & $\!(010)(011)\!$   & $\!(010)(111)\!$ & $\!(010)(100)\!$ & $\!(010)(001)\!$ & $\!(010)(101)\!$ \\
$\!\!\!(001)\!$ & $\!(001)(101)\!$  & $\!(001)(011)\!$ & $\!(001)(000)\!$   & $\!(001)(100)\!$ & $\!(001)(111)\!$ & $\!(001)(010)\!$ & $\!(001)(110)\!$ \\
$\!\!\!(110)\!$ & $\!(110)(010)\!$   & $\!(110)(100)\!$  &$\!(110)(111)\!$ & $\!(110)(011)\!$ & $\!(110)(000)\!$ & $\!(110)(101)\!$ & $\!(110)(001)\!$ \\
$\!\!\!(101)\!$ & $\!(101)(001)\!$   &  $\!(101)(111)\!$ & $\!(101)(100)\!$  & $\!(101)(000)\!$ & $\!(101)(011)\!$ & $\!(101)(110)\!$ & $\!(101)(010)\!$\\
$\!\!\!(011)\!$ &$\!(011)(111)\!$ &$\!(011)(001)\!$ &$\!(011)(010)\!$  & $\!(011)(110)\!$ & $\!(011)(101)\!$ & $\!(011)(000)\!$ & $\!(011)(100)\!$\\
$\!\!\!(111)\!$ &$\!(111)(011)\!$ &$\!(111)(101)\!$ &$\!(111)(110)\!$ & $\!(111)(010)\!$ & $\!(111)(001)\!$ & $\!(111)(100)\!$ & $\!(111)(000)\!$\\
\end{tabular}
\end{center}
\vspace{.4cm}
\caption{Quasi-uniform code constructed from $Z_2 \times Z_2 \times Z_2$.}
\end{table}}
\vspace{-5mm}

After applying the isomorphism in the encoding mapping, we get the final code as follows:

\begin{table}[H]
\begin{center}
\begin{tabular}{c|c|c|c|c|c|c|c|}
     & $G_1$ &  $G_2$ & $G_3$ & $G_4$ & $G_5$ & $G_6$ & $G_7$ \\
\hline
$\!\!\!(000)\!$ & $00$ &  $00$    & $00$ & $00$ & $00$ & $00$ & $00$ \\
$\!\!\!(100)\!$ & $00$   &  $10$  & $01$   & $10$ & $10$ & $01$ & $10$  \\
$\!\!(010)\!$ & $01$   &  $00$    & $10$   & $01$ & $10$ & $10$ & $01$ \\
$\!\!\!(001)\!$ & $10$  & $01$    & $00$   & $10$ & $01$ & $10$ & $11$ \\
$\!\!\!(110)\!$ & $01$   & $10$   & $11$ & $11$ & $00$ & $11$ & $11$ \\
$\!\!\!(101)\!$ & $10$   &  $11$  &  $01$  & $00$ & $11$ & $11$ & $01$\\
$\!\!\!(011)\!$ &$11$ &   $01$   & $10$  & $11$ & $11$ & $00$ & $10$\\
$\!\!\!(111)\!$ &$11$ &   $11$   & $11$ & $01$ & $01$ & $01$ & $00$\\
\end{tabular}
\end{center}
\vspace{.4cm}
\caption{Quasi-uniform code constructed from $Z_2 \times Z_2 \times Z_2$.}
\end{table}
\vspace{-5mm}
Note that the rows  corresponding to the basis vectors of $(\Z_2)^3$ forms the generator matrix of the code: 
\begin{table}[H]
\begin{center}
\begin{tabular}{|c|c|c|c|c|c|c|}
$00$   &  $10$  & $01$   & $10$ & $10$ & $01$ & $10$  \\
$01$   &  $00$    & $10$   & $01$ & $10$ & $10$ & $01$ \\
$10$  & $01$    & $00$   & $10$ & $01$ & $10$ & $11$ \\
\end{tabular}
\end{center}
\vspace{.4cm}
\caption{Generator matrix of the code}
\end{table}
\vspace{-5mm}
Here each codeword symbol is a block of two bits which corresponds to the cosets in the orginal encoding. However, we consider the length of the code as the total number of codeword symbols and  not the total number of bits.  



Some of the properties of the above code is mentioned below:
\begin{itemize}
\item The code is  linear of length 7.
\item The operation between codewords is defined component wise in $\Z_2 \times \Z_2.$
\item Although there is a generator matrix for the code, the codewords do not have an information part.
\item The cosets in different codeword symbols intersect trivially.
\end{itemize}

The decoding of a received vector is as follows:
\begin{itemize}
\item Suppose that the received codeword is $11|01|10|11|11|00|10$.
\item Compute the inverse image of the isomorphism $\gamma_i : G/G_i \rightarrow H_i$ (which is known for the user) for each codeword symbol $i=1,2,\cdots,7$ .
\item The inverse image obtained is $$(011)(111) |(011)(001) | (011)(010) | (011)(110) | (011)(101) | (011)(000) | (011)(100).$$
\item  Finally, the intersection of all cosets gives the information vector $011$.
\end{itemize}
This procedure can be simplified further, since the subgroups used for constructing the code have trivial pairwise intersections. Then the intersection of cosets from any two components  gives back the information vector. \\
That is; \textit{if the subgroups used have trivial pairwise intersections, then any two codeword symbols is enough to decode the whole information vector}.
\end{exm}

The concept implies that the code discussed earlier can also be interpreted as a batch code.  

Consider a request for \( t \) information symbols \((x_{i_1}, \ldots, x_{i_t})\). As noted in the code’s properties, information symbols are not stored directly in their original form within codewords. To address such a request, the entire information vector \((x_1, \ldots, x_k)\) of length \( k \) must first be decoded, after which the required symbol \( x_{i_1} \) (and others) can be retrieved. Since reconstructing the full information vector requires only two codeword symbols, the presence of \( t \) mutually disjoint recovery sets allows the independent retrieval of all \( t \) requested symbols \((x_{i_1}, \ldots, x_{i_t})\).  

Thus, the code in Example~\ref{exm:3} functions as a \((n, k, t) = (7, 3, 3)\) batch code. More specifically, it is a \((n, k, t, r) = (7, 3, 3, 2)\) batch code, where each recovery set has a maximum size $r=2$. Notably, the codeword symbol \( G_7 \) is unnecessary for handling \( t = 3 \) queries. Removing \( G_7 \) to shorten the code results in a \((6, 3, 3)\) batch code.  

\begin{defn}
A batch code constructed from subgroups of a finite group using quasi-unifrom construction method is called a \textit{quasi-uniform batch code}.
\end{defn}

A natural follow-up question is whether extending the code length could enable support for request sizes \( t > 3 \). This is indeed feasible but requires analyzing additional non-trivial subgroups of \((\mathbb{Z}_2)^3\) to structure the codeword symbols appropriately.

\begin{pro}
The group $(\Z_2)^3$ has exactly 7 subgroups of order 4.
\end{pro}

\begin{proof}
The group \((\mathbb{Z}_2)^3 = \{(a,b,c) \mid a,b,c \in \mathbb{Z}_2\}\) has subgroups of order \(4\) corresponding to \(2\)-dimensional subspaces. The number of such subspaces is calculated via the Gaussian binomial coefficient \cite{knuth}:
\[
\binom{3}{2}_2 = \frac{(2^3 - 1)(2^2 - 1)}{(2^2 - 1)(2^1 - 1)} = \frac{7 \cdot 3}{3 \cdot 1} = 7.
\]

The explicit subgroups are generated by pairs of linearly independent vectors:
\begin{enumerate}
    \item \(\langle (1,0,0), (0,1,0) \rangle = \{(0,0,0), (1,0,0), (0,1,0), (1,1,0)\}\)
    \item \(\langle (1,0,0), (0,0,1) \rangle = \{(0,0,0), (1,0,0), (0,0,1), (1,0,1)\}\)
    \item \(\langle (1,0,0), (0,1,1) \rangle = \{(0,0,0), (1,0,0), (0,1,1), (1,1,1)\}\)
    \item \(\langle (0,1,0), (0,0,1) \rangle = \{(0,0,0), (0,1,0), (0,0,1), (0,1,1)\}\)
    \item \(\langle (0,1,0), (1,0,1) \rangle = \{(0,0,0), (0,1,0), (1,0,1), (1,1,1)\}\)
    \item \(\langle (0,0,1), (1,1,0) \rangle = \{(0,0,0), (0,0,1), (1,1,0), (1,1,1)\}\)
    \item \(\langle (1,1,0), (0,1,1) \rangle = \{(0,0,0), (1,1,0), (0,1,1), (1,0,1)\}\)
\end{enumerate}

It can be verified easily  that the subgroups are distinct.
\end{proof}
The intersection of any two subgroups from the above set yields a subgroup of order 2. To enhance request capacity, a quasi-uniform code can be formulated by incorporating all subgroups of orders 2 and 4. For request recovery, it suffices to identify two subgroups whose intersection is trivial. As demonstrated in the subsequent analysis, 7 unique pairs (out of the 14 total subgroups) exhibit trivial pairwise intersections.

\begin{thm}
In \((\mathbb{Z}_2)^3\), there exist \(7\) distinct pairs \((H, S)\), where \(H\) is a subgroup of order \(4\) and \(S\) is a subgroup of order \(2\), such that \(H \cap S = \{(0,0,0)\}\). 
\end{thm}

\begin{proof}
We prove this result using Hall’s Marriage Theorem applied to a bipartite graph of trivially intersecting subgroups. Hall's Marriage Theorem states that
if $\mathcal{B} = (X \cup Y, E)$ is a bipartite graph with partitions \(X\) and \(Y\).  Then there exists a perfect matching in $\mathcal{B}$ if and only if for every subset \(A \subseteq X\), $|N(A)| \geq |A|$,
where \(N(A)\) denotes the set of neighbors of \(A\) in \(Y\).

Let \(G = (\mathbb{Z}_2)^3\). Consider the nontrivial subgroups of \(G\) of order 2 and 4. There are \(7\) subgroups of order \(2\) (1-dimensional subspaces), and \(7\) subgroups of order \(4\) (2-dimensional subspaces) in $G$.

Construct a bipartite graph \(\mathcal{G}\) with:
\begin{itemize}
    \item \textbf{Partition \(X\)}: The \(7\) order \(4\) subgroups.
    \item \textbf{Partition \(Y\)}: The \(7\) order \(2\) subgroups.
\end{itemize}

An edge connects \(H \in X\) to \(S \in Y\) if \(H \cap S = \{(0,0,0)\}\). 

We verify Hall’s condition for \(\mathcal{G}\):

\begin{itemize}
    \item Each \(H \in X\) contains \(3\) order \(2\) subgroups, so it trivially intersects \(7 - 3 = 4\) subgroups in \(Y\). Thus, \(\deg(H) = 4\).
    \item Each \(S \in Y\) is contained in \(3\) order \(4\) subgroups, so it trivially intersects \(7 - 3 = 4\) subgroups in \(X\). Thus, \(\deg(S) = 4\).
\end{itemize}
Therfore \(\mathcal{G}\) is a \(4\)-regular bipartite graph.

For any subset \(A \subseteq X\) with \(|A| = k\), the total edges incident to \(A\) are \(4k\). Since each \(S \in Y\) has degree \(4\), the number of neighbours \(|N(A)| \geq \frac{4k}{4} = k\). Thus:
\[
|N(A)| \geq |A| \quad \forall A \subseteq X.
\]

Then by Hall’s Marriage Theorem, \(\mathcal{G}\) contains a perfect matching. This guarantees \(7\) distinct pairs \((H, S)\) where \(H \cap S = \{(0,0,0)\}\).
\end{proof}
The theorem establishes that seven mutually disjoint recovery sets can be constructed by leveraging all non-trivial subgroups of orders 2 and 4 within the group $G=  (\Z_2)^3$. Crucially, these subgroups represent the complete set of non-trivial subgroups in $G$. By systematically incorporating these subgroups, we derive a quasi-uniform batch code capable of supporting query sizes as large as $t=7$. These insights are formalized in the following conclusion:

\begin{pro}
A $(14, 3, 7)$ linear quasi-uniform batch codes can be constructed using all  non-trivial subgroups of $(\Z_2)^3$.   
\end{pro}

Also note that the codeword symbols correponding to the 7  subgroups of size 4 live in $\Z_2$, not in $\Z_2 \times \Z_2$.

We extend the above construction of batch codes to $(\Z_2)^4$ in the next section.

\section{Quasi-uniform batch codes using $(\Z_2)^4$}
In order to construct a quasi-uniform batch code that accommodates the maximum number of queries utilizing the group \( G = (\Z_2)^4\), it is imperative to meticulously analyze the subgroup structure of \( G \).

\begin{pro}
\label{pro:2}
The number of distinct subgroups of order 2 (dimension 1) of $G=(\Z_2)^4$ is 15.
\end{pro}
\begin{proof}
We know that \((\mathbb{Z}_2)^4\) is a 4-dimensional vector space over the finite field \(\mathbb{Z}_2\) and its subgroups of order 2 correspond to 1-dimensional subspaces.

The number of \(k\)-dimensional subspaces in an \(n\)-dimensional vector space over \(\mathbb{Z}_q\) is given by the Guassian-binomial coefficient \cite{knuth}:
\begin{equation}
\label{eqn:guass}
\binom{n}{k}_q = \frac{(q^n - 1)(q^{n-1} - 1) \cdots (q^{n-k+1} - 1)}{(q^k - 1)(q^{k-1} - 1) \cdots (q - 1)}
\end{equation}

For our case: $q = 2, ~n = 4 $ and $ k = 1$.
Then $$\binom{4}{1}_2 = \frac{2^4 - 1}{2^1 - 1} = \frac{16 - 1}{2 - 1} = \frac{15}{1} = 15.$$

This matches the count of non-zero vectors in $(\mathbb{Z}_2)^4: \quad 2^4 - 1 = 15 $ and each non-zero vector generates a unique 1-dimensional subspace. That is, the number of distinct subgroups of order 2 in \((\mathbb{Z}_2)^4\) is 15.
\end{proof}

Similarly as the above result, subgroups of order 4 and order 8 of $G$ can be computed.

\begin{pro}
\label{pro:4}
The number of distinct subgroups of order 4 in \(G= (\mathbb{Z}_2)^4\) is 35.
\end{pro}
\begin{proof}
Subgroups of order 4 in $G$ correspond to 2-dimensional subspaces. Therefore, using Gaussian binomial coefficient (Equation \ref{eqn:guass}) for \(n = 4\), \(k = 2\), and \(q = 2\), we have the number of subgroups is given by
\[
\binom{4}{2}_2 = \frac{(2^4 - 1)(2^3 - 1)}{(2^2 - 1)(2^1 - 1)} = \frac{(16 - 1)(8 - 1)}{(4 - 1)(2 - 1)} = \frac{15 \cdot 7}{3 \cdot 1} = \frac{105}{3} = 35.
\]
\end{proof}

\begin{pro}
\label{pro:8}
The number of distinct subgroups of order $8$ in \(G= (\mathbb{Z}_2)^4\) is 15.
\end{pro}
\begin{proof}
Subgroups of order $8$ in $G$ correspond to 3-dimensional subspaces of $(\Z_2)^4$. Again, using Guassian binomial coeffiecient  (Equation \ref{eqn:guass}) for \(n = 4\), \(k = 3\), and \(q = 2\):
\begin{align*}
\binom{4}{3}_2 
&= \frac{(2^4 - 1)(2^3 - 1)(2^2 - 1)}{(2^3 - 1)(2^2 - 1)(2^1 - 1)} \\
&= \frac{(16 - 1)(8 - 1)(4 - 1)}{(8 - 1)(4 - 1)(2 - 1)} \\
&= \frac{15 \cdot 7 \cdot 3}{7 \cdot 3 \cdot 1} \\
&= \frac{315}{21} = 15.
\end{align*}

That is, the number of distinct subgroups of order 8 in \((\mathbb{Z}_2)^4\) is $15$.
\end{proof}
From Propositions \ref{pro:2} and \ref{pro:8}, we have the following conclusion:
\begin{rem}
There are 15 non-trivial subgroups each of order 2 and order 8 in $G=(\Z_2)^4$.
\end{rem}
Next Proposition shows that itis possible to pair up an order $2$ subgroup with an order $8$ subgroup such that their intersection is trivial. This idea is pivotel in determining disjoint recovery sets for batch code construction.
For proving the result, the following lemma is necessary.

\begin{lem}\label{lem:subgroups}
Let $G = (\mathbb{Z}_2)^n$ and let $K \subset G$ be a $k$-dimensional subspace. For any integer $m$ satisfying $k \leq m \leq n$, the number of $m$-dimensional subspaces of $G$ containing $K$ is given by the Gaussian binomial coefficient:
\[
\boxed{\binom{n-k}{m-k}_2}
\]
\end{lem}

\begin{proof}
Consider the quotient space $G/K$, which forms an $(n-k)$-dimensional vector space over $\mathbb{Z}_2$ by the dimension formula for quotient spaces \cite{Roman}. Then there exists a dimension-preserving bijection:
\[
\{\, M \subset G \mid K \subset M \,\} \longleftrightarrow \{\, \widetilde{M} \subset G/K \,\}
\]
where $\dim M = \dim \widetilde{M} + \dim K$ \cite{Roman}.

An $m$-dimensional subspace $M \subset G$ containing $K$ corresponds to an $(m-k)$-dimensional subspace $\widetilde{M} \subset G/K$. The number of such subspaces is given by the Gaussian binomial coefficient \cite{knuth}:
\[
\binom{n-k}{m-k}_2 = \prod_{i=0}^{m-k-1} \frac{2^{n-k} - 2^i}{2^{m-k} - 2^i} = \frac{(2^{n-k} - 1)(2^{n-k-1} - 1)\cdots(2^{n-m+1} - 1)}{(2^{m-k} - 1)(2^{m-k-1} - 1)\cdots(2^1 - 1)},
\] and the proof follows.
\end{proof}

\begin{pro}
\label{order28}
Let $H$ be an order 2 subgroup and let $S$ be an order 8 subgroup of $(\Z_2)^4$. There  exists 15 disjoint pairs of  $(H,S)$ where $H \cap S= \{(0,0,0,0)\}.$
\end{pro}
\begin{proof}
Let \(G = (\mathbb{Z}_2)^4\) be the additive group of 4-dimensional vectors over \(\mathbb{Z}_2\). We prove there exist 15 disjoint pairs \((H,S)\); where
\(H\) is a subgroup of order 2 (1-dimensional subspace),
\(S\) is a subgroup of order 8 (3-dimensional subspace),
 and \(H \cap S = \{(0,0,0,0)\}\) (i.e., \(H \not\subseteq S\)).

Define a bipartite graph \(\mathcal{G} = (X \cup Y, E)\) where:
\begin{itemize}
    \item \(X = \{\text{All order 2 subgroups } H \}\) (size \(|X| = 15\)),
    \item \(Y = \{\text{All order 8 subgroups } S \}\) (size \(|Y| = 15\)),
    \item An edge \(H \sim S\) exists if and only if \(H \cap S = \{(0,0,0,0)\}\).
\end{itemize}

Now for any \(H \in X\), the number of 3-dimensional subspaces \(S\) containing \(H\) is \(\binom{3}{2}_2 = 7\) from Lemma \ref{lem:subgroups}.
 Thus, \(H\) connects to \(15 - 7 = 8\) subspaces \(S \in Y\).
 
For any \(S \in Y\), the number of 1-dimensional subspaces \(H \subseteq S\) is \(7\) (since \(|S| = 8\) and each non-zero vector generates a 1-dimensional subspace). Thus, \(S\) connects to \(15 - 7 = 8\) subgroups \(H \in X\).
Hence, \(\mathcal{G}\) is an \(8\)-regular bipartite graph.

For any subset \(A \subseteq X\) with \(|A| = k\), the total edges incident to \(A\) are \(8k\). Since each \(S \in Y\) has degree \(8\), the number of neighbours \(|N(A)| \geq \frac{8k}{8} = k\). Thus:
\[
|N(A)| \geq |A| \quad \forall A \subseteq X.
\]
Therefore, by Hall’s marriage theorem, \(\mathcal{G}\) contains a perfect matching. That is, there exists  \(15\) disjoint pairs \((H,S)\) with \(H \cap S = \{(0,0,0,0)\}\). 
\end{proof}

To construct quasi-uniform codes supporting maximum number of queries using $G=(\Z_2)^4$ and it non-trivial subgroups, it remains to analyze the order 4 subgroups of $G$. For that, we need some additional framework.

\begin{defn}
A subspace $K$ intersects $ H$ trivially ($ H \cap K = \{\mathbf{0}\}$) if and only if $ K$ is a \textbf{complement} of $ H $ where $H$ and $K$ are subspaces of a vector space $V$.
\end{defn}


\begin{thm}\label{thm:complement}
Let \( V \) be an \( n \)-dimensional vector space over the finite field \( \mathbb{F}_q \), and let \( W \) be a \( k \)-dimensional subspace of \( V \). The number of complementary subspaces \( U \) of \( W \) (i.e., subspaces \( U \) such that \( V = W \oplus U \)) is given by:
\[
\boxed{q^{k(n - k)}}
\]
\end{thm}

\begin{proof}
Let \( \mathcal{B}_W = \{w_1, \dots, w_k\} \) be a basis for \( W \). Extend this to a basis \( \mathcal{B}_V = \{w_1, \dots, w_k, v_{k+1}, \dots, v_n\} \) for \( V \). We will show that there is a bijection between the set of complementary subspaces of \( W \) and the set of \( k \times (n - k) \) matrices over \( \mathbb{F}_q \), of which there are \( q^{k(n - k)} \).

\subsubsection*{Construction of complementary subspaces}
Every complementary subspace \( U \) must intersect \( W \) trivially and satisfy \( \dim U = n - k \). Define a candidate complementary subspace \( U_A \) for each \( k \times (n - k) \) matrix \( A = [a_{ij}] \) over \( \mathbb{F}_q \) as follows:
\[
U_A = \operatorname{span}\left\{\, \sum_{i=1}^k a_{i1}w_i + v_{k+1},\ \dots,\ \sum_{i=1}^k a_{i(n-k)}w_i + v_n \,\right\}.
\]
Equivalently, \( U_A \) is the column space of the block matrix:
\[
M_A = \begin{pmatrix} 
A \\ 
I_{n-k} 
\end{pmatrix},
\]
where \( I_{n-k} \) is the \( (n - k) \times (n - k) \) identity matrix. The columns of \( M_A \) are linearly independent because the lower block \( I_{n-k} \) ensures no nontrivial linear combination can vanish. Since \( \dim U_A = n - k \) and \( W \cap U_A = \{0\} \), \( U_A \) is indeed complementary to \( W \).

\subsubsection*{Injectivity}
Suppose \( U_A = U_{A'} \) for two matrices \( A \) and \( A' \). Then the columns of \( M_A \) and \( M_{A'} \) span the same subspace. This implies there exists an invertible matrix \( C \in \operatorname{GL}_{n-k}(\mathbb{F}_q) \) such that \( M_A C = M_{A'} \). Examining the lower block gives \( I_{n-k}C = I_{n-k} \), forcing \( C = I_{n-k} \). Consequently, \( A = A' \), proving injectivity.

\subsubsection*{Surjectivity}
Let \( U \) be an arbitrary complementary subspace of \( W \). Choose a basis \( \{u_1, \dots, u_{n-k}\} \) for \( U \). Since \( V = W \oplus U \), each \( u_j \) can be uniquely expressed as:
\[
u_j = \underbrace{\sum_{i=1}^k a_{ij}w_i}_{\text{component in } W} + \underbrace{v_{k+j}}_{\text{component in } \operatorname{span}\{v_{k+1}, \dots, v_n\}}.
\]
The coefficients \( a_{ij} \) define a matrix \( A \), and \( U = U_A \). Thus, every complementary subspace corresponds to some \( A \).

\subsubsection*{Counting}
The number of \( k \times (n - k) \) matrices over \( \mathbb{F}_q \) is \( q^{k(n - k)} \). By the bijection established above, this is also the number of complementary subspaces.
\end{proof}

\begin{exm}[Case \( n = 2, k = 1 \)]
Let \( V = \mathbb{F}_q^2 \) and \( W = \operatorname{span}\{(1, 0)\} \). Complementary subspaces are 1-dimensional subspaces not equal to \( W \). These are of the form \( \operatorname{span}\{(a, 1)\} \) for \( a \in \mathbb{F}_q \), yielding exactly \( q \) complements, consistent with \( q^{1(2-1)} = q \).
\end{exm}

\begin{exm}[Case \( n = 3, k = 1 \)]
Let \( V = \mathbb{F}_q^3 \) and \( W = \operatorname{span}\{(1, 0, 0)\} \). A complementary subspace \( U \) has a basis of the form:
\[
(a, 1, 0),\ (b, 0, 1),
\]
for \( a, b \in \mathbb{F}_q \). There are \( q^2 \) choices, matching \( q^{1(3-1)} = q^2 \).
\end{exm}




From the above theorem we have the following corollary:
\begin{cor}
\label{cor:2}
For a fixed 2-dimensional subspace $K$ of $ (\mathbb{Z}_2)^4  $, there are 16,  2-dimensional subspaces \( H \) such that \( K \cap H = \{0,0,0,0\} \). 
\end{cor}
\begin{proof}
Let \( K \subset G \) with \( \dim(K) = 2 \) and \( \dim(G) = 4 \). Here $q=2$ as well.  From Theorem \ref{thm:complement},  the number of complementary subspaces is given by
$
2^{2 \cdot (4 - 2)} = 2^4 = 16.
$
\end{proof}

The above result says that, for any subgroup $K$ of $G=(\mathbb{Z}_2)^4 $ of order 4, there exists16 other subgroups of order 4, which are trivially intersecting with $K$. Therefore, if  we construct a graph with whole 35 order 4 subgroups as vertices and trivially intersecting subgroups forms edges, the resulting graph is a simple 16-regular graph.

Let us analyze the edge-connectivity of  such graphs.

\begin{defn}
    The \textbf{edge connectivity} $\lambda(G)$ of $G$ is the minimum number of edges whose removal disconnects $G$.
\end{defn}

The following results can be used to compute the number of pairs on non-intersecting subgroups of order 4.

\begin{thm}\label{thm:aut}
The automorphism group \(\text{Aut}(G)\) of \(G=(\mathbb{Z}_2)^4\) is isomorphic to \(\text{GL}(4, 2)\).
\end{thm}

\begin{proof}
A group automorphism of \(G\) is a bijective map \(T: G \to G\) that preserves the group operation:
\[
T(\mathbf{u} + \mathbf{v}) = T(\mathbf{u}) + T(\mathbf{v}), \quad \forall \mathbf{u}, \mathbf{v} \in G.
\]
Since \(G\) is a vector space over \(\mathbb{Z}_2\), automorphisms must also preserve scalar multiplication:
\[
T(c\mathbf{u}) = cT(\mathbf{u}), \quad \forall c \in \mathbb{Z}_2, \, \mathbf{u} \in G.
\]
Thus, every automorphism of \(G\) is a \textit{linear transformation}.

Let \(\mathcal{B} = \{\mathbf{e}_1, \mathbf{e}_2, \mathbf{e}_3, \mathbf{e}_4\}\) be the standard basis for \(G\). Any linear transformation \(T\) can be represented by a \(4 \times 4\) matrix \(A_T\) over \(\mathbb{Z}_2\), where:
\[
T(\mathbf{e}_j) = \sum_{i=1}^4 a_{ij} \mathbf{e}_i, \quad a_{ij} \in \mathbb{Z}_2, ~j=1,\cdots,4.
\]
For \(T\) to be invertible, \(A_T\) must have a non-zero determinant in \(\mathbb{Z}_2\). The set of all such matrices forms \(\text{GL}(4, 2)\). 

To show there are no non-linear automorphisms, observe that any additive bijection preserving \(\mathbb{Z}_2\)-scalar multiplication must be linear (see \cite{Roman}). Thus, \(\text{Aut}(G) \cong \text{GL}(4, 2)\).
\end{proof}

\begin{defn}[Vertex transitivity]
A graph \(\Gamma = (V, E)\) is \textbf{vertex-transitive} if its automorphism group \(\text{Aut}(\Gamma)\) acts \textit{transitively} on its vertex set \(V\). Formally, for any two vertices \(u, v \in V\), there exists an automorphism \(\sigma \in \text{Aut}(\Gamma)\) such that:
\[
\sigma(u) = v.
\]
Equivalently, the automorphism group \(\text{Aut}(\Gamma)\) induces a single orbit on \(V\) under its action.
\end{defn}

\noindent \textbf{Implications:}
\begin{itemize}
    \item All vertices in a vertex-transitive graph are \textit{structurally indistinguishable}; the graph ``looks the same" from any vertex.
    \item Vertex-transitive graphs are necessarily \textit{regular} (all vertices have the same degree).
\end{itemize}

\begin{thm}\label{thm:vt}
Let \(\Gamma\) be the graph whose vertices are 2-dimensional subspaces of \(G=(\Z_2)^4\), with two subspaces adjacent if they intersect trivially. Then \(\Gamma\) is vertex-transitive.
\end{thm}

\begin{proof}
The automorphism group \(\text{GL}(4, 2)\) acts on the set of 2-dimensional subspaces via:
\[
T \cdot H = \{T(\mathbf{v}) \mid \mathbf{v} \in H\}, \quad T \in \text{GL}(4, 2), \, H \leq G.
\]
By the fundamental transitivity property of \(\text{GL}(n, q)\) on $k$-dimensional subspaces, for any two 2-dimensional subspaces \(H_1, H_2\), there exists \(T \in \text{GL}(4, 2)\) such that \(T(H_1) = H_2\). 

To see that adjacency is preserved, suppose \(H_1 \cap H_2 = \{(0,0,0,0)\}\). Then:
\[
T(H_1) \cap T(H_2) = T(H_1 \cap H_2) = T(\{(0,0,0,0)\}) = \{(0,0,0,0)\}.
\]
Thus, \(T\) maps adjacent vertices to adjacent vertices. By Theorem \ref{thm:aut}, \(\text{GL}(4, 2)\) acts transitively on the vertices of \(\Gamma\), making \(\Gamma\) vertex-transitive.
\end{proof}

\begin{thm}[Edge connectivity theorem]
\label{watkin}
Let \(\Gamma = (V, E)\) be a connected, vertex-transitive, non-bipartite graph with regularity \(k\). Then, the edge connectivity \(\lambda(\Gamma)\) satisfies \(\lambda(\Gamma) = k\).
\end{thm}

\begin{proof}
Assume for contradiction that \(\lambda(\Gamma) < k\). Let \(S\) be a minimal edge cut with \(|S| = \lambda(\Gamma)\), partitioning \(\Gamma\) into components \(A\) and \(B\). By vertex transitivity, every vertex in \(A\) (resp. \(B\)) is incident to exactly \(t\) edges in \(S\), where \(t = |S| / |A| = |S| / |B|\). Since \(|S| < k\), we have \(t < k\), so each vertex retains \(k - t \geq 1\) edges within its component.

To see that \(A\) and \(B\) remain connected, observe that the minimality of \(S\) ensures no proper subset of \(S\) disconnects \(\Gamma\). If \(A\) were disconnected into subcomponents \(A_1\) and \(A_2\), a smaller edge cut \(S' \subset S\) could disconnect \(A_1\) from \(A_2\), contradicting the minimality of \(S\). Thus, \(A\) and \(B\) are connected subgraphs.

Menger’s theorem [ \cite{diestel}, Theorem 3.3.6] states that, a graph is \(k\)-edge-connected if and only if any two vertices are connected by \(k\) edge-disjoint paths. Therefore any two vertices \(u \in A\) and \(v \in B\) have \(k\) edge-disjoint paths in \(\Gamma\). However, removing \(|S| < k\) edges disconnects \(u\) and \(v\), which is impossible. This contradiction forces \(\lambda(\Gamma) \geq k\). Since \(\lambda(\Gamma) \leq k\) trivially, we conclude \(\lambda(\Gamma) = k\).
\end{proof}

\begin{cor}
\label{edgeconn}
\(\Gamma\) is a 16-regular, vertex-transitive, non-bipartite graph having 35 vertices with edge connectivity 16.
\end{cor}
\begin{proof}
Vertex transitivity follows from Theorem \ref{thm:vt} and hence the graph is connected (Since each component has at least 17 vertices and components containing 17 and 18 vertices is not possible since the graph is vertex transitive).  Regularity (degree 16) was shown in Corollary \ref{cor:2} by counting complementary subspaces. It can be verified easily that the graph contains odd cycles and hence it is not bi-partite. Now by Theorem \ref{watkin}, a connected, vertex-transitive, non-bipartite graph has edge connectivity equal to its regularity. That is $\lambda(\Gamma)=16$.
\end{proof}

\begin{thm}[Tutte's theorem \cite{Tutte1947}]
A graph $G$ has a perfect matching if and only if for every $U \subseteq V(G)$, the number of odd-sized components in $G - U$ is at most $|U|$.
\end{thm}

\begin{lem}[Edge bound for odd components]
\label{lemma:edge_bound}
Let $G$ be a $k$-regular graph with $\lambda(G) \geq k-1$. For any $U \subseteq V(G)$ and any odd component $O_i$ in $G - U$, the number of edges $e(O_i, U)$ between $O_i$ and $U$ satisfies $e(O_i, U) \geq k$.
\end{lem}

\begin{proof}
Since $\lambda(G) \geq k-1$, any edge cut separating $O_i$ from $U$ has at least $k-1$ edges. If $k$ is even, $e(O_i, U) =k|O_i| - 2|E(O_i)|$, is even, and so $e(O_i, U) \geq k$ (Here $E(O_i)$is the number of edges within $O_i)$. If $k$ is odd, $e(O_i, U)$ is odd (as $k|O_i|$ is odd and $2|E(O_i)|$ is even), so $e(O_i, U) \geq k$. 
\end{proof}

\begin{thm}[Even vertex case]
\label{thm:even}
Every $k$-regular graph with an even number of vertices and $\lambda(G) \geq k-1$ contains a perfect matching.
\end{thm}
\begin{proof}
Assume $G$ has no perfect matching. By Tutte's Theorem, there exists $U \subseteq V(G)$ such that the number of odd components $c(G - U) > |U|$. By Lemma \ref{lemma:edge_bound}, each odd component $O_i$ satisfies $e(O_i, U) \geq k$. The total edges from $U$ to all odd components is $\geq k \cdot c(G - U)$. Since $G$ is $k$-regular, this total is $\leq k|U|$. Thus:
\[
k \cdot c(G - U) \leq k|U| \implies c(G - U) \leq |U|,
\]
contradicting $c(G - U) > |U|$. Hence, $G$ has a perfect matching. $\square$
\end{proof}

\begin{thm}[Odd vertex case]
\label{thm:odd}
Every $k$-regular graph with an odd number of vertices and $\lambda(G) \geq k-1$ has a maximum matching of size $\left\lfloor \frac{|V|}{2} \right\rfloor$.
\end{thm}

\begin{proof}
Let $\text{def}(G) = \max_{U \subseteq V} (c(G - U) - |U|)$. By Tutte's Theorem, the maximum matching size is $\frac{1}{2}(|V| - \text{def}(G))$, which is known as the Tutte-Berge formula \cite{tutteberge}. Assume $\text{def}(G) \geq 2$. Then, there exists $U$ with $c(G - U) \geq |U| + 2$. By Lemma \ref{lemma:edge_bound}, each odd component contributes $\geq k$ edges to $U$. Thus:
\[
k(|U| + 2) \leq \sum e(O_i, U) \leq k|U| \implies 2k \leq 0,
\]
a contradiction. Hence, $\text{def}(G) \leq 1$, and the maximum matching size is $\left\lfloor \frac{|V|}{2} \right\rfloor$.
\end{proof}

\begin{pro}
\label{order4}
The graph $\Gamma$ formed using subgroups of $G= (\Z_2)^4$ of order 4 as vertices and edges formed between vertices corresponding to subgroups with trivial intersections, has a maximum matching of size 17.
\end{pro}

\begin{proof}
From Corollary \ref{edgeconn} and Theorem \ref{thm:odd}, there exists a maximum matching in $\Gamma$ containing  $\left\lfloor \frac{35}{2} \right\rfloor = 17 $.
\end{proof}

In the group $(\Z_2)^4$, there are 17 distinct non-overlapping pairs of subgroups, each of order 4. This property implies that when designing a batch code with non-trivial subgroups, the subgroups of order 4 can also accommodate 17 queries. By integrating all such non-trivial subgroups, a batch code with parameters $(65,4,32)$ is achieved, where $n=65, k=4$ and $t=32$.

\begin{pro}
Consider $G= (\Z_2)^4$ and all of its non-trivial subgroups. There exists a $(65,4,32)$ quasi-uniform batch code.
\end{pro} 

\begin{proof}
The group $G = (\mathbb{Z}_2)^4$ is known to contain 15 subgroups of order 2, 35 subgroups of order 4, and 15 subgroups of order 8. A quasi-uniform code $C$ is derived from these subgroups, resulting in parameters $n = 65$ and $k = 4$, where each information symbol corresponds to a vector in $G$. The total code size is $16$. 

By Propositions \ref{order28} and \ref{order4}, there are 15 unique disjoint subgroup pairs of orders 2 and 8 with trivial intersections, as well as 17 distinct disjoint subgroup pairs of order 4 that similarly intersect trivially. For any single request, these pairs enable recovery through one of the available options, collectively providing $32$ recovery set pairs in $C$. Consequently, $C$ functions as a $(n, k, t) = (65, 4, 32)$ batch code.
\end{proof}

\section{Batch codes using $(\Z_2)^k$}
This section extends the subgroup-based framework from $(\mathbb{Z}_2)^4$ to a general group $G = (\mathbb{Z}_2)^k$, where $k \in \mathbb{N}$. The analysis involves two key steps: First, determining the count of nontrivial $m$-dimensional subgroups for integers $1 \leq m \leq k-1$. Second, identifying complementary subspaces for each subgroup and enumerating distinct pairs of such subspaces (subgroups), which play a critical role in building quasi-uniform batch codes.  

To address parity-specific properties, the computations are structured differently for odd and even values of $k$.

\begin{rem}\label{def:gaussian}
The number of \(m\)-dimensional subspaces of \((\mathbb{Z}_2)^k\) is given by the Gassian-binomial coefficient:
\[
\binom{k}{m}_2 = \frac{(2^k - 1)(2^{k-1} - 1) \cdots (2^{k - m + 1} - 1)}{(2^m - 1)(2^{m - 1} - 1) \cdots (2^1 - 1)}.
\]
\end{rem}


\begin{rem}
\label{sgpsodd}
The total number of nontrivial subgroups of  \(G=(\mathbb{Z}_2)^k\) is $\sum_{m=1}^{k-1} \binom{k}{m}_2$ and the number of subspaces of $G$ of dimension $k-m$ is also same as the number of $m$-dimensional subspaces for $m=1,\cdots, \frac{k-1}{2}$ because of the symmetry of the Guassian-binomial coefficient. 
\end{rem}

\noindent \textit{Reminder}: Subspaces \(H\) and \(S\) of \((\mathbb{Z}_2)^k\) are \textbf{complementary} if \(H \cap S = \{\bf{0}\}\) and \(H + S = (\mathbb{Z}_2)^k\).


We need the following Lemmas to prove a major result of this section.

\begin{lem}[Counting complements]\label{lem:complement}
Every \(m\)-dimensional subspace \(H \subseteq (\mathbb{Z}_2)^k\) has exactly \(2^{m(k- m)}\) complementary \((k- m)\)-dimensional subspaces in 
$(\Z_2)^k$.
\end{lem}
This follows directly from Theorem \ref{thm:complement}.


\begin{lem}[Uniform degree property]\label{lem:degree}
In the bipartite graph \(G(X, Y)\) where \(X\) and \(Y\) are sets of $m$ and \((k - m)\)-dimensional subspaces, every vertex has degree \(2^{m(k - m)}\).
\end{lem}

\begin{proof}
By Lemma \ref{lem:complement}, each \(H \in X\) connects to \(2^{m(k - m)}\) subspaces \(S \in Y\), and vice versa.
\end{proof}
\begin{thm}\label{thm:main}
Let $k \in \mathbb{N}$ and \(1 \leq m \leq k-1\). The number of distinct pairs \((H, S)\) of \(m\)-dimensional and \((k - m)\)-dimensional subspaces of \((\mathbb{Z}_2)^k\) with \(H \cap S = \{\bf{0}\}\) is \(\binom{k}{m}_2\).
\end{thm}

\begin{proof}
Construct a bipartite graph \(G(X, Y)\) with  
\(X = \{m\text{-dimensional subspaces}\}\), \(Y = \{(k - m)\text{-dimensional subspaces}\}\) and edges exist between \(H \in X\) and \(S \in Y\) if \(H \cap S = \{\bf{0}\}\).  

For any subset \(A \subseteq X\),  it follows that $N(A)| \geq |A|$ similarly as in the proof of Proposition \ref{order28}. That is, Hall’s condition holds for all subsets \(A \subseteq X\).  

Then by the Hall's marriage theorem, \(G(X, Y)\) admits a perfect matching.
Therfore, the number of distinct pairs \((H, S)\) of \(m\)-dimensional and \((k - m)\)-dimensional subspaces of \((\mathbb{Z}_2)^k\) with \(H \cap S = \{\bf{0}\}\) is \(\binom{k}{m}_2\).
\end{proof}
This result establishes a bijection between the set of \(m\)-dimensional and \((k - m)\)-dimensional subspaces of \((\mathbb{Z}_2)^k\) having trivial intersection. 

For odd $k$, regular bipartite graphs corresponding to subspaces of dimension $m$ can be constructed for all $m$ within the defined parameters. However, when $k$ is even, no such bipartite graph exists for subspaces of dimension $\frac{k}{2}$. This fundamental disparity necessitates distinct treatments for odd and even $k$.  

\begin{pro}[$k$ is an odd positive integer]
Let $G= (\mathbb{Z}_2)^k$, $k$ is an odd integer. Then there exists a quasi-uniform batch code of length $n= \sum_{m=1}^{\frac{k-1}{2}}2. \binom{k}{m}_2$, dimension  $k$, and $t= \sum_{m=1}^{\frac{k-1}{2}}\binom{k}{m}_2.$
\end{pro}

\begin{proof}
The result follows directly using Remark \ref{sgpsodd}, Theorem \ref{thm:main} and by the method of construction of quasi-uniform batch codes from Section \ref{sec:construction}.
\end{proof}

\subsection{Batch codes using  $(\Z_2)^k$, when $k$ is even}

\begin{pro}
Let $k$ be an even positive  integer. Then $k=2s$ where $s\in \mathbb{N}$. The number of subgroups of order $2^s$ in the group $(\mathbb{Z}_2)^k$ is given by the Gaussian binomial coefficient
\[
\binom{2s}{s}_2 = \prod_{i=0}^{s-1} \frac{2^{2s-i} - 1}{2^{s-i} - 1}.
\]
\end{pro}

\begin{proof}
The group $(\mathbb{Z}_2)^k$ carries the natural structure of a $k$-dimensional vector space over the finite field $\mathbb{Z}_2$. In this context, subgroups of $(\mathbb{Z}_2)^k$ correspond precisely to vector subspaces, and the order of a subgroup is directly determined by the dimension of its corresponding subspace. Specifically, a subgroup of order $2^s$ corresponds to an $s$-dimensional subspace.

The number of $s$-dimensional subspaces of a $2s$-dimensional vector space over $\mathbb{Z}_2$ is enumerated by the Gaussian binomial coefficient:
\[
\binom{2s}{s}_2 = \frac{(2^{2s} - 1)(2^{2s-1} - 1)\cdots(2^{s+1} - 1)}{(2^s - 1)(2^{s-1} - 1)\cdots(2^1 - 1)},
\]
which simplifies to the product formula
\[
\prod_{i=0}^{s-1} \frac{2^{2s-i} - 1}{2^{s-i} - 1}.
\]
\end{proof}

\begin{thm}
\label{thm:reg}
Let $H$ be a fixed $s$-dimensional subgroup of $(\mathbb{Z}_2)^{k=2s}$. The number of $s$-dimensional subgroups $K$ of $(\mathbb{Z}_2)^{2s}$ satisfying $H \cap K = \{\bf{0}\}$ is $2^{s^2}$. 
\end{thm}
\begin{proof}
The proof follows directly from Theorem \ref{thm:complement}.
\end{proof}



Let $\Gamma$ denotes the graph defined as follows: the vertices are all $s$-dimensional subgroups of the group $G=(\Z_2)^{2s}$, and two vertices are adjacent if the corresponding subgroups intersect trivially. According to Theorem \ref{thm:reg}, this graph is  $2^{s^2}$-regular
and has $\binom{2s}{s}_2$ vertices. 

We now demonstrate that $\Gamma$ possesses the properties of being vertex-transitive, connected, and non-bipartite.

\begin{thm}
Let \( G = (\mathbb{Z}_2)^{2s} \). The graph \(\Gamma\), defined on \(s\)-dimensional subgroups of \(G\) with adjacency given by trivial intersections, is non-bipartite.
\end{thm}

\begin{proof}
We construct an explicit triangle (3-cycle) in \(\Gamma\) using properties of complementary subspaces and linear maps.

\medskip\noindent
\textbf{Step 1: Constructing complementary subspaces}\\
Fix a decomposition \( G = H \oplus K \), where \(H\) and \(K\) are \(s\)-dimensional subspaces. Let \(\phi: H \to K\) be a linear isomorphism. 
Define:
\[
L = \{ h + \phi(h) \mid h \in H \}.
\]
This set \(L\) is closed under addition and scalar multiplication:
\begin{itemize}
    \item For \(h_1, h_2 \in H\), \( (h_1 + \phi(h_1)) + (h_2 + \phi(h_2)) = (h_1 + h_2) + \phi(h_1 + h_2) \in L \).
    \item For \(c \in \mathbb{Z}_2\), \( c(h + \phi(h)) = ch + \phi(ch) \in L \).
\end{itemize}
Thus, \(L\) is a subspace of \(G\). 

\medskip\noindent
\textbf{Step 2: Dimension of \(L\)}\\
The map \(\psi: H \to L\) defined by \(\psi(h) = h + \phi(h)\) is a linear isomorphism:
\begin{itemize}
    \item \textit{Injectivity:} If \(\psi(h_1) = \psi(h_2)\), then \(h_1 - h_2 = \phi(h_2) - \phi(h_1)\). Since \(H \cap K = \{0\}\), this implies \(h_1 = h_2\).
    \item \textit{Surjectivity:} Every element of \(L\) is of the form \(h + \phi(h)\) for some \(h \in H\).
\end{itemize}
Hence, \(\dim(L) = \dim(H) = n\).

\medskip\noindent
\textbf{Step 3: Trivial intersections}\\
We verify pairwise trivial intersections:
\begin{itemize}
    \item \(H \cap K = \{0\}\) by construction.
    \item \(H \cap L = \{0\}\): If \(h = h' + \phi(h')\) for \(h, h' \in H\), then \(\phi(h') = h - h' \in H \cap K = \{0\}\), so \(h = h' = 0\).
    \item \(K \cap L = \{0\}\): If \(k = h + \phi(h)\) for \(k \in K\), then \(h = k - \phi(h) \in H \cap K = \{0\}\), so \(k = 0\).
\end{itemize}

\medskip\noindent
\textbf{Step 4: Non-bipartiteness}\\
The subspaces \(H\), \(K\), and \(L\) form a triangle in \(\Gamma\):
\[
H \sim K \sim L \sim H.
\]
A graph containing a triangle (odd cycle) cannot be bipartite. Thus, \(\Gamma\) is non-bipartite.
\end{proof}

\begin{thm}
Let \( G = (\mathbb{Z}_2)^{2s} \), and let \(\Gamma\) be the graph whose vertices are \(s\)-dimensional subspaces of \(G\), with edges connecting two subspaces if their intersection is trivial. Then \(\Gamma\) is vertex transitive.
\end{thm}

\begin{proof}
To show vertex transitivity, we use the group \(\mathrm{GL}(2s, \mathbb{Z}_2)\) of invertible linear transformations on \(G\). We prove that this group acts transitively on the vertices of \(\Gamma\) while preserving adjacency.

\textit{Transitivity on subspaces:}
Let \(H\) and \(K\) be two $s$-dimensional subspaces of \(G\). We construct an invertible linear map \(g \in \mathrm{GL}(2s, \mathbb{Z}_2)\) such that \(g(H) = K\) as follows:

\begin{enumerate}
    \item \textbf{Basis extension for \(H\):} 
    Choose a basis \(\mathcal{B}_H = \{v_1, \ldots, v_s\}\) for \(H\). Extend this to a full basis of \(G\):
    \[
    \mathcal{B}_G = \{v_1, \ldots, v_s, w_1, \ldots, w_s\}.
    \]
    
    \item \textbf{Basis extension for \(K\):} 
    Similarly, choose a basis \(\mathcal{B}_K = \{u_1, \ldots, u_s\}\) for \(K\). Extend this to another full basis of \(G\):
    \[
    \mathcal{B}_G' = \{u_1, \ldots, u_s, z_1, \ldots, z_s\}.
    \]
    
    \item \textbf{Constructing \(g\):} 
    Define the linear map \(g: G \to G\) by:
    \[
    g(v_i) = u_i \quad \text{and} \quad g(w_j) = z_j \quad \text{for all } 1 \leq i, j \leq n.
    \]
    Since \(g\) maps the basis \(\mathcal{B}_G\) to the basis \(\mathcal{B}_G'\), it is invertible (\(g \in \mathrm{GL}(2s, \mathbb{Z}_2)\)). By construction, \(g(H) = K\).
\end{enumerate}

\textit{Structure preservation :}
The map \(g\) preserves the subspace structure:
\begin{itemize}
    \item \textbf{Dimension:} Since \(g\) maps a basis of \(H\) to a basis of \(K\), \(\dim(g(H)) = \dim(H) = s\).
    
    \item \textbf{Intersections:} For any subspaces \(H, K \subset G\),
    \[
    g(H \cap K) = g(H) \cap g(K).
    \]
    If \(H \cap K = \{\bf{0}\}\), then \(g(H) \cap g(K) = g(\{\bf{0}\}) = \{\bf{0}\}\). Thus, adjacency in \(\Gamma\) is preserved.
\end{itemize}

Therefore, for any two \(s\)-dimensional subspaces \(H\) and \(K\), the transformation \(g \in \mathrm{GL}(2s, \mathbb{Z}_2)\) maps \(H\) to \(K\) while preserving edges. Therefore, \(\Gamma\) is vertex transitive. \end{proof}


Determining the connectivity of \(\Gamma\) presents significant challenges. Suppose \(\Gamma\) is disconnected. In that case, every connected component must be vertex-transitive and mutually isomorphic, as the vertex transitivity of \(\Gamma\) ensures automorphisms map components to isomorphic counterparts. The total count of $s$-dimensional subspaces of $G$ is \(\binom{2s}{s}_2\), which is odd (as both numerator and denominator in the binomial coefficient are odd). Consequently, each component contains an odd number of vertices, exceeding \(2^{s^2}\) (a consequence of regularity). By Theorem \ref{watkin}, the edge connectivity of every component equals \(2^{s^2}\). This leads to the following conclusion regarding maximum matchings in \(\Gamma\).

\begin{thm}
\label{thm:maxmatch}
Let \( G = (\mathbb{Z}_2)^{2s} \), and let \(\Gamma\) be the graph whose vertices are \(s\)-dimensional subspaces of \(G\), with edges connecting two subspaces if their intersection is trivial. If $\Gamma$ has $r$, $r\ge 1$ components, then the maximum matching of $\Gamma$ has size $\sum_{i=1}^r \left \lfloor \frac{|V_i|}{2}  \right \rfloor$.
\end{thm}
\begin{proof}
Since $\Gamma$ is vertex transitive, all of its components are vertex transitive and isomorphic. Therefore, if $\Gamma$ has more than one component, each component is a $2^{s^2}$-regular, non bipartite graph with odd number of vertices with edge connectivity $2^{s^2}$. Then by Theorem \ref{thm:odd}, the  maximum matching size of $i$-th component is $\left \lfloor \frac{|V_i|}{2}  \right \rfloor$ and hence the result follows.
\end{proof}

Using the above results, we can construct quasi-unifrm batch codes using subgroups of $G=(\Z_2)^k$, where $k$ is even, as follows:

\begin{pro}
Let $G= (\mathbb{Z}_2)^k$, $k$ is an even positive integer. Then there exists a quasi-uniform batch code of length $n= \sum_{m=1}^{k-1} \binom{k}{m}_2$, dimension  $k$ and $$t= \sum_{m=1}^{\frac{k}{2}-1}\binom{k}{m}_2+  \sum_{i=1}^r \left \lfloor \frac{|V_i|}{2}  \right \rfloor$$ where $r$ is the number of components in the graph corresponding to the subspaces of dimension $\frac{k}{2}$ as in the previous theorem.
\end{pro}
\begin{proof} 
There are $\binom{k}{\frac{k}{2}}_2$ subgroups of dimension $\frac{k}{2}$. From Theorem  \ref{thm:maxmatch}, there exists $\sum_{i=1}^r \left \lfloor \frac{|V_i|}{2}  \right \rfloor$ pairs of subgroups of dimension $\frac{k}{2}$ where $\sum_{i=1}^r |V_i|= \binom{k}{\frac{k}{2}}_2$.

For all other non-trivial subgroups of dimension $m \in \{1, \cdots, k-1\}\setminus {\frac{k}{2}}$, the number of recovery sets of size 2 can be computed as $\sum_{m=1}^{\frac{k}{2}-1}\binom{k}{m}_2$ using similar ideas in  Theorem \ref{thm:main}.  

Combining all, we have the required pairs of disjoint recovery subgroups.
\end{proof}

Next result computes the optimum length for a batchcode in terms of the batch size $t$.

\begin{thm}
Given any batch of requests $(i_1,i_2,\cdots, i_t)$ of information symbols of size $t$, the optimal $(n,k,t)$-batch code $C$ has length $\bf{n\le  2t-r}$ where $1\le r \le \mbox{min}\{ t, k\}$ is the number of distinct information symbols in the batch of requests.
\end{thm}

\begin{proof}
Given the batch of requests $(i_1,i_2,\cdots, i_t)$, we have to construct  an $(n,k,t)$ batch code. \\
Let $(i_1,i_2, \cdots, i_r)$ be the distinct information symbols in the batch and the remaining symbols are repetitions of them. If the codewords contain all informations symbols as coded symbols, the $r$ distinct symbols are directly served. For serving one of the remaining requests,  combinations at least 2 coded symbols are needed and all recovery sets are disjoint. Hence, we must have at least $2(t-r)$ coded symbols are necessary to serve the $t-r$ queries remaining. In total $n\le 2(t-r)+r = 2t-r$.
\end{proof}

Note that the length $n$ is depending on $t$ and $r$, not on $k$.

\begin{rem}
The quasi-uniform batch codes constructed above are near optimal for a given $t$ since the length $n \approx 2t$.
\end{rem} 

\section{Quasi-uniform codes as locally repairable codes}
As demonstrated in prior sections, quasi-uniform codes exhibit structural properties that align with their utility as \textit{batch codes} \cite{dimakis}. Notably, their reconstruction mechanism for information symbols inherently relies on intermediate recovery of coded symbols. If the decoding process terminates after reconstructing these coded symbols (rather than progressing to information symbols), the code effectively functions as a \textit{Locally Repairable Code (LRC)} \cite{vitali}. This duality arises from a fundamental distinction between batch codes and LRCs: the former prioritize direct recovery of \textit{information symbols}, while the latter focus on efficient repair of \textit{coded symbols} via localized parity checks. 

A closer examination of this relationship reveals that the quasi-uniform framework naturally supports LRC-like resilience. Specifically, these codes admit recovery from $t$ erasures among coded symbols while maintaining an \textit{availability} parameter of 2, as defined in the LRC literature \cite{vitali}. However, the interplay between quasi-uniformity and LRC parameters (e.g., locality, minimum distance) requires further algebraic analysis beyond our current scope. 

\section{Conclusion}
This work presents a new method to construct batch codes for distributed storage systems by adapting quasi-uniform codes—originally built using algebraic groups—to prioritize efficient recovery of requested data batches. Focusing on $2$-groups, we leverage their subgroup structures to design codes that balance storage loads effectively. The symmetry of these groups ensures the codes can handle large request batches while naturally supporting features of Locally Repairable Codes (LRCs) and Private Information Retrieval (PIR) codes. By blending tools from linear algebra, combinatorics, and group theory, this approach demonstrates how abstract mathematical structures can solve practical coding challenges.

While our construction highlights the potential of group theory in code design, open questions remain. Future work could extend this framework to p-groups (for primes $p>2$) and compare performance with traditional batch codes to identify trade-offs. Additionally, the deep interplay between subgroup analysis and coding theory revealed here suggests a broader principle: algebraic regularity can simplify code engineering. This connection not only advances storage solutions but also enriches the dialogue between pure mathematics and real-world applications.

\section{Acknowledgement}
The author thanks Vitaly Skachek for insightful suggestions.
 
%
%

\medskip

\medskip

\end{document}